\documentclass{llncs}
\usepackage[utf8x]{inputenc}
\usepackage{ucs}
\usepackage{mfirstuc}
\usepackage{xspace}
\usepackage{amsmath}
\usepackage{graphicx}
\usepackage[frozencache,cachedir=.]{minted}
\usepackage{url}

\setminted[python]{%
  fontsize=\fontsize{9pt}{1.2}
}

\newcommand{\furl}[1]{\footnote{\url{#1}}}
\newcommand{\toolname}[1]{\capitalisewords{#1}}
\newcommand{\Haydi}{\toolname{haydi}}

\newcommand{\Atoms}{\mathcal{A}}
\newcommand{\Objs}{\mathcal{O}}
\newcommand{\Perms}{\mathcal{P}}

\newcommand{\Cn}{\mathcal{C}}

\newcommand{\uset}{\mathsf{uset}}
\newcommand{\duset}{\mathsf{uset_\downarrow}}
\newcommand{\getatoms}{\mathsf{atoms}}
\newcommand{\cf}{\mathsf{cf}}

\newcommand{\parent}{\mathsf{parent}}

\title{Haydi: Rapid Prototyping and Combinatorial Objects}
\titlerunning{Haydi}

\author{%
  Stanislav B\"{o}hm \and Jakub Ber\'{a}nek \and Martin \v{S}urkovsk\'{y}
  \thanks{Authors were supported by grants of GACR 15-13784S and SGS
    No. SP2017/82, V\v{S}B -- Technical University of Ostrava, Czech Republic.
    Computational resources for testing on a super computer were provided by
    IT4Innovations as project OPEN-8-26.}
}
\authorrunning{S. B\"{o}hm \and J. Ber\'{a}nek \and M. \v{S}urkovsk\'{y}}

\institute{%
  V\v{S}B -- Technical University of Ostrava,\\17. listopadu 2172/15, 708 00 Ostrava, Czech Republic\\
  \email{\{stanislav.bohm, jakub.beranek.st, martin.surkovsky\}@vsb.cz}
}

\begin{document}

\maketitle

\begin{abstract}
	\Haydi{} (http://haydi.readthedocs.io) is a framework for generating discrete structures. It provides a way to define a structure from basic building blocks and then enumerate all elements, all non-isomorphic elements, or generate random elements in the structure. \Haydi{} is designed as a tool for rapid prototyping. It is implemented as a pure Python package and supports execution in distributed environments. The goal of this paper is to give the overall picture of Haydi together with a formal definition for the case of generating canonical forms.
\end{abstract}


\section{Introduction}

The concept of rapid prototyping helps in verifying the feasibility of an initial idea and reject the bad ones fast. In mathematical world, there are tools like \toolname{Matlab}, \toolname{SageMath}, or \toolname{R} that allow one to build a working prototype and evaluate an idea quickly.
This paper is focused on the field of combinatorial objects and provides a prototyping tool
that allows to check claims on small instances by search over relevant objects. \Haydi{} (Haystack Diver) is an open-source Python package that provides an easy way of describing such structures by composing basic building blocks (e.g. Cartesian product, mappings) and then enumerating all elements, all non-isomorphic elements, or generating random elements.

The main design goal is to build a tool that is \emph{simple} to use, since building prototypes have to be cheap and fast.
There has been an attempt to build a \emph{flexible} tool that describes various structures and reduces limitations for the user.
The \emph{reasonable performance} of the solution is also important, but it has a lower priority than the first two goals.

To fulfill these goals, \Haydi{} has been built as a Python package.
Python is a well-known programming language and is commonly used as a prototyping language that provides a high degree of flexibility.
Since \Haydi{} is written purely in Python, it is compatible with PyPy\furl{https://pypy.org/} -- a fast Python implementation with JIT compiler.
Moreover, \Haydi{} is designed to transparently utilize a cluster of computers to provide a better performance without sacrificing simplicity or flexibility.
The distributed execution is built over \toolname{Dask/distributed}\furl{https://github.com/dask/distributed} and it was tested on Salomon cluster\furl{https://docs.it4i.cz/salomon/introduction/}. 

The goal of this paper is to give the overall picture of \Haydi{} together with
a formal definition for the case of generating canonical forms.
More detailed and programmer-oriented text can be found in the user guide\furl{http://haydi.readthedocs.io/}.
\Haydi{} is released as an open source project at \url{https://github.com/spirali/haydi} under MIT license.

The original motivation for the tool was to investigate hard instances for equivalence of deterministic push-down automata (DPDAs).
We have released a data set containing non-equivalent normed DPDAs~\cite{ndpda}.

The paper starts with two motivation examples in Section~\ref{s:examples} followed by covering related works in Section~\ref{s:relworks}.
Section~\ref{s:arch} introduces the architecture of \Haydi{}. Section~\ref{s:cnfs} shows a theoretical framework for generating canonical forms.
Section~\ref{s:distcomp} covers a basic usage of distributed computations and used optimizations. The last section shows performance measurements.

\section{Examples}\label{s:examples}

To give an impression of how \Haydi{} works, basic usage of \Haydi{} is
demonstrated on two examples.
The first one is a generator for directed graphs and the second one
is a generator of finite state automata for the reset word problem.

\subsection{Example: Directed graphs}\label{s:ex-dg}

In this example, our goal is to generate directed graphs with $n$ nodes.
Our first task is to describe the structure itself: we represent a graph as a set of edges,
where an edge is a pair of two (possibly the same) nodes. For the simplicity of outputs,
we are going to generate graphs on two nodes. However, this can be simply changed by editing a single constant,
namely the number 2 on the second line in following code:

\begin{minted}{python}
>>> import haydi as hd
>>> nodes = hd.USet(2, "n")  # A two-element set with elements {n0, n1}
>>> graphs = hd.Subsets(nodes * nodes)  # Subsets of a cartesian product
\end{minted}

The first line just imports \Haydi{} package. The second one creates a set of nodes, namely a set of two ``unlabeled'' elements.
The first argument is the number of elements, the second one is the prefix of each element name.
The exact meaning of \emph{USet} will be discussed further in the paper.
For now, it just creates a
set with elements without any additional quality, the elements of this set can
be freely relabeled. In this example, it provides us with the standard graph isomorphism.
The third line constructs a collection of all graphs on two nodes, 
in a mathematical notation it could be written as ``$\mathsf{graphs} = \mathcal{P}(\mathsf{nodes} \times \mathsf{nodes})$''.

With this definition, we can now iterate all graphs:

\begin{minted}{python}
>>> list(graphs.iterate())
[{}, {(n0, n0)}, {(n0, n0), (n0, n1)}, {(n0, n0), (n0, n1), (n1, n0)},
# ... 3 lines removed ...
n1)}, {(n1, n0)}, {(n1, n0), (n1, n1)}, {(n1, n1)}]
\end{minted}

or iterate in a way in which we can see only one graph per isomorphic class:

\begin{minted}{python}
>>> list(graphs.cnfs())  # cnfs = canonical forms
[{}, {(n0, n0)}, {(n0, n0), (n1, n1)}, {(n0, n0), (n0, n1)},
{(n0, n0), (n0, n1), (n1, n1)}, {(n0, n0), (n0, n1), (n1, n0)},
{(n0, n0), (n0, n1), (n1, n0), (n1, n1)}, {(n0, n0), (n1, n0)},
{(n0, n1)}, {(n0, n1), (n1, n0)}]
\end{minted}

or generate random instances (3 instances in this case):

\begin{minted}{python}
>>> list(graphs.generate(3))
[{(n1, n0)}, {(n1, n1), (n0, n0)}, {(n0, n1), (n1, n0)}]
\end{minted}

\Haydi{} supports standard operations such as \emph{map}, \emph{filter}, and \emph{reduce}.
The following example shows how to define graphs without loops, i.e. graphs such that for all edges $(a, b)$ hold that $a \neq b$:

\begin{minted}{python}
>>> no_loops = graphs.filter(lambda g: all(a!=b for (a,b) in g.to_set()))
\end{minted}

All these constructions can be transparently evaluated as a pipeline distributed across a cluster.
\Haydi{} uses \toolname{Dask/distributed} for distributing tasks, the following code assumes that \toolname{dask/distributed} server
runs at \texttt{hostname:1234}:

\begin{minted}{python}
# Initialization
>>> from haydi import DistributedContext
>>> context = DistributedContext("hostname", 1234)

# Run a pipeline
>>> graphs.iterate().run(ctx=context)
\end{minted}

\subsection{Example: Reset words}
\label{s:ex-cc}

A \emph{reset word} is a word that sends all states of a given finite automaton to a unique state.
The following example generates automata and computes the length of a minimal reset word.
It can be used for verifying the \v{C}ern\'{y} conjecture on bounded instances.
The conjecture states that the length of a minimal reset word is bounded by $(n-1)^2$ where $n$ is the number of states of the automaton~\cite{cerny1964,Volkov2008}.

First, we describe deterministic automata by their transition functions (a mapping from a pair of state and symbol to a new state).
In the following code, \verb|n_states| is the number of states and \verb|n_symbols| is the size of the alphabet.
We use \verb|USet| even for the alphabet, since we do not care about the meaning of particular symbols, we just need to distinguish them.

\begin{minted}{python}
# set of states q0, q1, ..., q_{n_states-1}
>>> states = hd.USet(n_states, "q")
# set of symbols a0, ..., a_{a_symbols-1}
>>> alphabet = hd.USet(n_symbols, "a")

# All mappings (states * alphabet) -> states
>>> delta = hd.Mappings(states * alphabet, states)
\end{minted}

Now we can create a pipeline that goes through all the automata of the given size (one per an isomorphic class) and finds the maximal length among minimal reset words:

\begin{minted}{python}
>>> pipeline = delta.cnfs().map(check_automaton).max(size=1)
>>> result = pipeline.run()

>>> print ("The maximal length of a minimal reset word for an "
...        "automaton with {} states and {} symbols is {}.".
...        format(n_states, n_symbols, result[0]))
\end{minted}

The function \verb|check_automaton| takes an automaton (as a transition function) and returns the length of the minimal reset word, or 0 when there is no such a word. It is just a simple breadth-first search on sets of states. The function is listed in Appendix~\ref{A:check-automaton}.

\section{Related works}\label{s:relworks}

Many complex software frameworks are designed for rapid checking mathematical ideas, for example \toolname{Maple}, \toolname{Matlab}, \toolname{SageMath}.
Most of them also contain a package for combinatorial structures, e.g. \toolname{Combinatorics} in \toolname{SageMath}\furl{http://doc.sagemath.org/html/en/reference/combinat/sage/combinat/tutorial.html},
\toolname{combstruct} in \toolname{Maple}\furl{https://www.maplesoft.com/support/help/Maple/view.aspx?path=combstruct}.

From the perspective of the mentioned tools, \Haydi{} is a small single-purpose package.
But as far as we know, there is no other tool that allows building structures by composition, searching only one structure of each isomorphism class as well as offering simple execution in distributed environment.

Tools focused on the generation of specific structures are on the other side of
the spectrum.
One example is \toolname{Nauty}~\cite{McKayGeng} that contains \toolname{Geng}
for generating graphs, another ones are generators for parity games in \toolname{PGSolver}~\cite{PGSolver} or automata generator for \toolname{SageMath}~\cite{ASM}. These tools provide highly optimized generators for a given structure.

\section{Architecture}\label{s:arch}

\Haydi{} is a Python package for rapid prototyping of generators for discrete structures. 
The main two components are \emph{domains} and \emph{pipelines}.
The former is dedicated to defining structures and the latter executes an operation over domains.
In this section, both domains and pipelines are introduced. Parts that are
related to generating canonical forms are omitted. This is covered separately in Section~\ref{s:cnfs}.

\subsection{Domains}

The basic structure in \Haydi{} is a \emph{domain} that represents an unordered collection of (Python) objects.
On abstract level, domains can be viewed as countable sets with some implementation details.
The basic operations with the domains are iterations through their elements and generating a random element. 
Domains are composable, i.e., more complex domains can be created from simpler ones.

There are six \emph{elementary} domains: Range (a range of integers), Values (a domain of explicitly listed Python objects), Boolean (a domain containing \verb|True| and \verb|False|), and NoneDomain (a domain containing only one element: \verb|None|). Examples are shown in Figure~\ref{f:elementary}.
There are also domains \verb|USet| and \verb|CnfValues|; their description is
postponed to Section~\ref{s:cnfs}, since it is necessary to develop a theory to explain their purpose.

\begin{figure}
\begin{minted}{python}
>>> import haydi as hd

>>> hd.Range(4)  # Domain of four integers
<Range size=4 {0, 1, 2, 3}>

>>> hd.Values(["Haystack", "diver"])
<Values size=2 {'Haystack', 'diver'}>
\end{minted}
\caption{Examples of elementary domains}
\label{f:elementary}
\end{figure}

New domains can be created by composing existing ones or applying a transformation. There are the following compositions: 
\emph{Cartesian product, sequences, subsets, mappings}, and \emph{join}. Examples are shown in Figure~\ref{f:domain-composition}, more details can be found in the user guide. There are two transformations \emph{map} and \emph{filter} with the standard meaning. Examples are shown 
in Figure~\ref{f:domain-trs}.

\begin{figure}
\begin{minted}{python}
>>> import haydi as hd
>>> a = hd.Range(2)
>>> b = hd.Values(("a", "b", "c"))

>>> hd.Product((a, b))  # Cartesian product
<Product size=6 {(0, 'a'), (0, 'b'), (0, 'c'), (1, 'a'), ...}>

>>> a * b  # Same as above
<Product size=6 {(0, 'a'), (0, 'b'), (0, 'c'), (1, 'a'), ...}>

>>> hd.Subsets(a)  # Subsets of 'a'
<Subsets size=4 {{}, {0}, {0, 1}, {1}}>

>>> hd.Mappings(a, a)  # Mappings from 'a' to 'a'
<Mappings size=4 {{0: 0; 1: 0}, {0: 0; 1: 1}, {0: 1; 1: 0}, ...}>

>>> hd.Sequences(a, 3)  # Sequences of length 3 over 'a'
<Sequences size=8 {(0, 0, 0), (0, 0, 1), (0, 1, 0), (0, 1, 1), ...}>

>>> hd.Join((a, b))  # Join 'a' and 'b', can be also written as 'a + b'
<Join size=5 {0, 1, 'a', 'b', 'c'}>
\end{minted}
\caption{Examples of domain compositions}
\label{f:domain-composition}
\end{figure}

\begin{figure}
\begin{minted}{python}
>>> a = hd.Range(5)

>>> a.map(lambda x: x * 10)
<MapTransformation size=5 {0, 10, 20, 30, 40}>
\end{minted}
\caption{Examples of domain compositions}
\label{f:domain-trs}
\end{figure}

\subsection{Pipeline}

Domains in the previous section describe a set of elements.
Pipelines provide a way how to work with elements in these sets.
Generally, a pipeline provides methods for generating and iterating elements and optionally 
applying simple ``map \& reduce'' transformations.  

The pipeline creates a stream of elements from a domain by one of the three methods.
We can apply transformations on elements in the stream. The pipeline ends by a reducing action.
The schema is shown in Figure~\ref{f:pipeline}. The pipeline consists of:

\begin{figure}
	\centering
	\includegraphics[width=\textwidth]{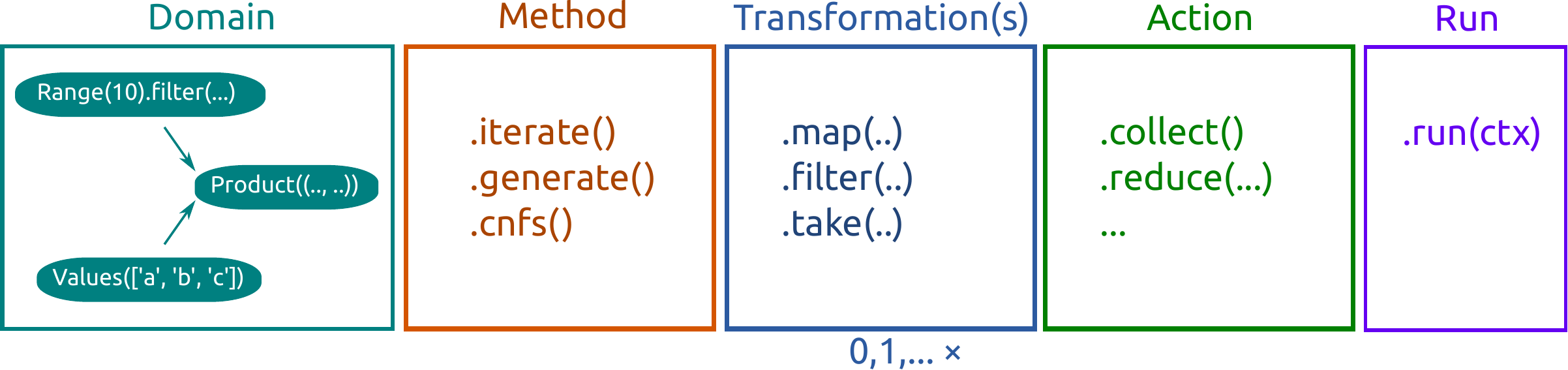}
	\caption{The pipeline schema}
	\label{f:pipeline}
\end{figure}

 
\textbf{Method}\quad It specifies how to take elements from the domain into the stream. Haydi provides three options: \verb|iterate()|, \verb|generate(n)|, and \verb|cnfs()|.
Method\\* \verb|iterarate()| iterates all elements of a given domain, \verb|generate(n)| creates $n$ random elements of the domain (by default with the uniform distribution over all elements),
and \verb|cnfs()| iterates over canonical forms (Section~\ref{s:cnfs}).

\textbf{Transformations}\quad Transformation modifies/filters elements in a stream.\\* There are three pipeline transformations:
\verb|map(fn)| -- applies the function \verb|fn| on each element that goes through the pipeline,
\verb|filter(fn)| -- filters elements in the pipeline according to the provided function,
\verb|take(count)| -- takes only first \verb|count| elements from the stream. 
The reason why transformations on domains and in pipeline are distinguished is described in \url{https://haydi.readthedocs.io/en/latest/pipeline.html#transformations}.

\textbf{Actions}\quad Action is a final operation on a stream of elements. 
For example there are: \verb|collect()| -- creates a list form of the stream,
\verb|reduce(fn)| -- applies binary operation on elements of the stream,
\verb|max()| -- takes maximal elements in the stream.

\textbf{run()}\quad The previous operations declare the pipeline, which is an immutable representation of a computational graph. The \verb|run()| method actually executes the pipeline. The optional \verb|ctx| (context) parameter specifies how should the computation be performed (serially or in a distributed way on a cluster).

The examples of pipelines are shown in Figure~\ref{f:pipelines}.
Not all parts of a pipeline have to be specified, if some of them are missing, defaults are used; the default method is \verb|iterate()|
and the default action is \verb|collect()|.

\begin{figure}
\begin{minted}{python}
>>> domain = hd.Range(5) * hd.Range(3)

# Iterate all elemenets and collect them
>>> domain.iterate().collect().run()
[(0, 0), (0, 1), (0, 2), (1, 0), (1, 1), (1, 2), (2, 0),
 (2, 1), (2, 2), (3, 0), (3, 1), (3, 2), (4, 0), (4, 1), (4, 2)]
 
# The same as above, since iterate() and collect() is default
>>> domain.run()
[(0, 0), (0, 1), (0, 2), (1, 0), (1, 1), (1, 2), (2, 0),
 (2, 1), (2, 2), (3, 0), (3, 1), (3, 2), (4, 0), (4, 1), (4, 2)]
 
# Generate three elements 
>>> domain.generate(3).run()
[(3, 2), (4, 0), (1, 2)]

# Take elements that are maximal in first component
>>> domain.max(lambda x: x[0]).run()
[(4, 0), (4, 1), (4, 2)]

\end{minted}
\caption{Examples of pipelines}
\label{f:pipelines}
\end{figure}

\section{Generating canonical forms}\label{s:cnfs}

In many cases, when we want to verify a property of a discrete structure, 
we are not interested in the names of the elements in the structure.
For example, in the case of graphs we usually want to see only one graph for each isomorphic class.
Another example can be finite-state automata; in many cases we are not especially interested in names of states
and actual symbols in the input alphabet. For example, in \v{C}ern\'{y} conjecture, the minimal length of reset words is not changed when the alphabet 
is permuted. On the other hand, symbols in some other problems may have special meanings and we cannot freely interchange them.

\Haydi{} introduces \verb|haydi.USet| as a simple but expressive mechanism for describing what permutations we are interested in.
It serves to define partitions of atomic objects. Each partition creates a
set of atomic objects that can be freely interchanged with one another; 
we call these partitions \emph{Unlabeled sets}.
They then establish semantics to what structure should be preserved when isomorphisms on various structures are defined.

\Haydi{} allows to iterate a domain in a way where we see only one element for
each isomorphic class. It is implemented as an iteration through \emph{canonical forms} (CNFS).

In this section, a simple theoretical framework is built. It gives a formal background to this feature.
Paragraphs starting with ``\emph{Abstract:}'' are meant as part of a theoretical description. 
Paragraphs starting with ``\emph{Haydi:}'' describe the implementation of the framework in \Haydi{}. 

\emph{Abstract:}
Let $\Atoms$ be a set of all atomic objects whose structure is not investigated any further.
The set of objects $\Objs$ is the minimal set with the following properties:

\begin{itemize}
\item $\Atoms \subseteq \Objs$ \quad (atoms)
\item If $o_i \in \Objs$ for $i \in \{1, 2, \dots, n \}$ then $\{o_1, o_2, \dots, o_n\} \in \Objs$ \quad (finite sets)
\item If $o_i \in \Objs$ for $i \in \{1, 2, \dots, n \}$ then $(o_1, o_2, \dots, o_n) \in \Objs$ \quad (finite sequences)
\end{itemize}

It is assumed that type of each object (atom, sequence, and set) can always be determined.
Therefore, it is assumed that sequences and sets are not contained in atoms. 

Function $\getatoms: \Objs \rightarrow 2^\Atoms$ that returns the atoms
contained in an object is defined as follows:

\begin{itemize}
	\item $\getatoms(o) = \{o\}$ if $o \in \Atoms$
	\item $\getatoms(o) = \bigcup_{i \in \{ 1, \dots, n \}} \getatoms(o_i)$ if 
	$o = \{o_1, \dots, o_n\}$ or $o = (o_1, \dots, o_n)$
\end{itemize}

\emph{\Haydi{}:}
The used Python instantiation of the definitions is the following: $\Atoms$ contains \verb|None|, \verb|True|, \verb|False|, all instances of types \verb|int| (integers) and \verb|str| (strings) and instances of the class \verb|haydi.Atom|.
Except for the last one, they are Python built-in objects; the last one is related to unlabeled sets and will be explained later.
Sequences in $\Objs$ are identified with Python tuples, sets with \verb|haydi.Set| (analogous to standard \verb|set|).
\Haydi{} also contains the type \verb|haydi.Map| (analogous to \verb|dict|) for
mappings. In the theoretical framework, mappings were not explicitly
distinguished, as they can be considered sets of pairs.
For sets and maps, standard python objects are not directly used for performance reasons\footnote{
Built-in classes \texttt{set} and \texttt{dict} are optimized for lookups;
however, Haydi needs fast comparison methods as will be seen later. \texttt{haydi.Set} and \texttt{haydi.Map} are stored in a sorted state to enable this.}; however, both \verb|haydi.Set| and \verb|haydi.Map| can be directly transformed into their standard Python counter-parts.

Note: Generally, domains in \Haydi{} may contain any Python object; however, domains that support iterating over CNFS 
impose some restrictions that will be shown later. Since the theoretical framework 
is built just for CNFS, its formalization to Python is mapped in a way that
respects these limitations from the beginning. For this reason, $\Objs$ is not
identified with all Python objects.
The Python incarnation of $\Objs$ is called \emph{basic objects}.

Now we established an isomorphism between objects. 
We define that two objects are isomorphic if they can be obtained
one from another by permuting its atoms.
To control permutations, partitioning of atoms is introduced and permutation of
atoms is allowed only within its ``own'' class.
These classes are defined through $\uset$, that is an abbreviation of ``unlabeled set''.


\emph{Abstract:}
Let us fix a function $\uset: \Atoms \rightarrow 2^{\Atoms}$ in the following way:
\begin{itemize}
\item $\forall a \in \Atoms: a \in \uset(a)$
\item $\forall a, b \in \Atoms: \uset(a) = \uset(b) \lor \uset(a) \cap \uset(b) = \emptyset$
\end{itemize}

Obviously $\uset$ partitions $\Atoms$ into disjoint classes.

Let $\Perms$ be a set of all bijective functions from $\Atoms$ to $\Atoms$ such that for each $\pi \in \Perms$ holds $\forall a \in \Atoms: \pi(a) \in \uset(a)$.

Applying $\pi \in \Perms$ to an object $o \in \Objs$ (written as $o^\pi$) is defined as follows:
\begin{itemize}
\item $o^\pi = \pi(o)$ if $o \in \Atoms$
\item $o^\pi = \{o_1^\pi, \dots, o_n^\pi\}$ if $\{o_1, \dots, o_n\} = o$
\item $o^\pi = (o_1^\pi, \dots, o_n^\pi)$ if $(o_1, \dots, o_n) = o$
\end{itemize}

Let $o_1, o_2 \in \Objs$ then $o_1$ and $o_2$ are \emph{isomorphic} (written as $o_1 \equiv o_2$) if there exists $\pi \in \Perms$ such that $o_1 = o_2^\pi$.

\emph{\Haydi{}:}
All integers, strings, \verb|None|, \verb|True|, and \verb|False| have a singleton unlabeled set, i.e., $\uset(a) = \{a\}$.
Therefore, all objects that contain only these atoms always form their own ``private'' isomorphic class. 
For example: \verb|("abc", 1)| cannot be isomorphic to anything else since string \verb|"abc"| and integer \verb|1| cannot be replaced,
because for each $\pi \in \Perms$ holds $\pi(\mathtt{"abc"}) = \mathtt{"abc"}$ and $\pi(\mathtt{1}) = \mathtt{1}$.

The only way to create a non-singleton unlabeled set is to use domain \verb|haydi.USet| (Unlabled set) that creates a set of atoms belonging to the same unlabeled set; in other words, if $X$ is created by \verb|haydi.USet| then for each $o \in X$ holds $\uset(o) = X$.

\begin{minted}{python}
>>> a = hd.USet(3, "a")
>>> list(a)
[a0, a1, a2]	
\end{minted}

The first argument is the size of the set, and the second one is the name of the set that has only informative character.
The name is also used as the prefix of element names, again without any semantical meaning.
Elements of \texttt{USet} are instances of \texttt{haydi.Atom} that is a wrapper over an integer and a reference to the \verb|USet| that contains them.

Method \verb|haydi.is_isomorphic| takes two objects and returns \verb|True| iff the objects are isomorphic according to our definition.
Several examples are shown in Figure~\ref{f:isom}.

\begin{figure}
\begin{minted}{python}
>>> a0, a1, a2 = hd.USet(3, "a")
>>> b0, b1 = hd.USet(2, "b")
>>> hd.is_isomorphic(a0, a2)
True
>>> hd.is_isomorphic(b1, b0)
True
>>> hd.is_isomorphic(a0, b0)
False
>>> hd.is_isomorphic((a0, b0), (a2, b1))
True
>>> hd.is_isomorphic((a0, a0), (a0, a2))
False
\end{minted}
\caption{Isomorphism examples}
\label{f:isom}
\end{figure}

\subsection{Canonical forms}

\Haydi{} implements an iteration over canonical forms as a way to obtain exactly one element for each isomorphic class.
We define a canonical form as the smallest element from the isomorphic class according to a fixed ordering. 

\emph{Abstract:} We fix a binary relation $\leq$ for the rest of the section 
such that $\Objs$ is well-ordered under $\leq$. As usual, we write $o_1 < o_2$ if $o_1 \leq o_2$ and $o_1 \neq o_2$.
A \emph{canonical form} of an object $o$ is $\cf(o) = \min \{ o' \in \Objs \mid o \equiv o' \}$.
We denote $\Cn = \{ o \in \Objs \mid o = \cf(o) \}$ as a set of all canonical forms.


\emph{\Haydi{}:} Canonical forms can be generated by calling \verb|cnfs()| on a domain.
Figure~\ref{f:cnfs} shows simple examples of generating CNFS. In case 1, we have only two 
results \verb|(a0, a0)| and \verb|(a0, a1)|; the former represents a pair with
the same two values
and the latter represents a pair of two different values. Obviously, we cannot get one from the other by applying any permutation,
and all other elements of Cartesian product \verb|a * a| can be obtained by permutations. This fact is independent of the size of \verb|a| (as it has at least two elements).
In case 2, the result is two elements, since we cannot permute elements from different usets. The third case shows 
canonical forms of a power set of \texttt{a}, as we see there is exactly one canonical form for each size of sets.

\begin{figure}
\begin{minted}{python}
>>> a = hd.USet(3, "a")
>>> b = hd.USet(2, "b")
	
>>> list((a * a).cnfs())  # 1
[(a0, a0), (a0, a1)]
	
>>> list((a + b).cnfs())  # 2
[a0, b0]
	
>>> list(hd.Subsets(a).cnfs())  # 3
[{}, {a0}, {a0, a1}, {a0, a1, a2}]
\end{minted}
\caption{CNFS examples}
\label{f:cnfs}
\end{figure}

Generating CNFS is limited in \Haydi{} to \emph{strict} domains that have the following features:

\begin{enumerate}
\item A strict domain contains only basic objects (defined at the beginning of this section).
\item A strict domain is closed under isomorphism.
\end{enumerate}

The first limitation comes from the need of ordering. The standard comparison method \verb|__eq__| is
not sufficient since it may change between executions. To ensure deterministic
canonical forms\footnote{
In Python 2, instances of different types are generally unequal, and they are ordered consistently but arbitrarily.
Switching to Python 3 does not help us, since comparing incompatible types throws an error (e.g. \texttt{3 < (1, 2)}), hence
standard comparison cannot serve as ordering that we need for basic objects.
},
\Haydi{} defines \texttt{haydi.compare} method. This  method is responsible for deterministic 
comparison of basic objects and provides some additional properties that are
explained later. The second condition ensures that canonical forms represent all
elements of a domain. Usually these conditions do not present a practical
limitation. Elementary domains except for \verb|haydi.Values| are always strict and standard compositions preserve strictness. Elementary domain \verb|haydi.CnfsValues| allows to define a (strict) domain through canonical elements,
hence it serves as a substitute for \verb|haydi.Values| in a case when a strict
domain from explicitly listed elements is needed.

\subsection{The Algorithm}

This section describes implementation of the algorithm that generates
canonical forms. A na\"{i}ve approach would be to iterate over all elements and filter out non-canonical ones.
\Haydi{} avoids the na\"{i}ve approach and makes the generation of canonical forms more efficient. It constructs new elements from smaller ones in a depth-first search manner. 
On each level, relevant extensions of the object are explored, and non-canonical
ones are pruned.
The used approach guarantees that all canonical forms are generated, and each will be generated exactly once, 
hence the already generated elements do not need to be remembered (except the current branch in a building tree).

This approach was already used in many applications and extracted into an abstract framework (e.g. ~\cite{McKay98}).
The main goal of this section is to show correctness of the approach used in
\Haydi{} and not to give an abstract framework for generating canonical elements, since it was done before. However, the goal is
not to generate a specific kind of structures, but provide a framework for their
describing, therefore, a rather abstract approach must still be used.

Let us note that the algorithm is not dealing here with efficiency of deciding whether a given element is in a canonical form. In our use cases, most elements are relatively
small, hence all relevant permutations are checked during checking the canonicity
of an element. Therefore, the implementation in \Haydi{} is quite straightforward.
It exploits some direct consequences of Proposition~\ref{p:no-gap} that allow
the algorithm to reduce the set of relevant permutations and in some cases immediately claim non-canonicity.

The used approach is based on the following two propositions.
The first says that an object cannot be canonical if it  contains ``gaps'' in
atoms occurring in the object. The second shows that new elements can only be
constructed from existing canonical forms and still all of them are reached.
 
At the beginning, let us introduce some properties of the ordering given by
\texttt{haydi.compare} which allow the propositions to be established.
On the abstract level, the following properties for ordering $\leq$ are assumed where $o, o' \in \Objs$:

\begin{itemize}
\item Tuples of the same length are lexicographically ordered.
\item If $o = \{ o_1, \dots, o_n \}$ where $o_1 < \dots < o_n$ and $o' = \{ o'_1, \dots, o'_n \}$ where $o'_1 < \dots < o'_n$ then
$o \leq o'$ if $(o_1, \dots, o_n) \leq (o'_1, \dots, o'_n)$.
\end{itemize}

A set $X \subseteq \Atoms$ \emph{contains a gap} if there exists $a \in X$ such that there is $a' \in \Atoms \setminus X$ and $a' \in \duset(a)$ where $\duset(a) = \{ a' \in \uset(a) \mid a' < a  \}$.

\begin{proposition}\label{p:no-gap}
If $o \in \Objs$ and $\getatoms(o)$ contains a gap, then $o$ is not a canonical form. 	
\end{proposition}

Proposition~\ref{p:no-gap} is a direct consequence of the following claim:

\begin{proposition}
If $o \in O$ and $a', a \in \Atoms$ such that $a' \in \duset(a), a \in \getatoms(o), a' \notin \getatoms(o)$
and $\pi \in \Perms$ is a permutation that only swaps $a$ and $a'$
then $o^\pi < o$.
\end{proposition}

\begin{proof}
The proposition is proved by induction on the structure of $o$; let $a, a', \pi$ be as in the statement of the proposition:
If $o$ is an atom then directly $o = a, o^\pi = a'$ and $a' < a$ from assumptions.
Now assume that $o = (o_1, \dots, o_n)$ and the proposition holds for all $o_i, i \in \{1, \dots, n\}$.
From assumptions we get that each $o_i$ does not contain $a'$ and there is the
minimal index $f \in \{1, \dots, n\}$ such that $o_f$ contains $a$. Hence $o_i =
o_i^\pi$ for all $i \in \{1, \dots, f-1\}$ and $o_f^\pi < o$ by the induction
assumption. Since tuples are lexicographically ordered it follows that $o^\pi < o$. Similar ideas apply also for sets.
\qed
\end{proof}

Let us define function $\parent : \Objs \rightarrow \Objs \cup \{ \bot \}$ (where $\bot$ is a fresh symbol) that gives rise to a search tree. The function returns
a ``smaller'' object from which the object may be constructed. The function returns $\bot$ for ``ground'' objects (atoms, empty tuples/sets). 

\[   
\parent(o) = 
\begin{cases}
\bot & \text{if } o \in \Atoms \cup \{ (), \{\} \} \\
(o_1, \dots, o_n) & \text{if } o = (o_1, \dots, o_{n+1}) \\
o \setminus \{ x \} & \text{if } o \text{ is a non-empty set and } x = \max o \\  
\end{cases}
\]



\begin{proposition}
	\label{p:parent}
	For each $o \in \Objs$ holds:
	\begin{enumerate}
	\item Exists $n \in \{1, 2, \dots \}$ such that $\parent^n(o) = \bot$.
	\item If $o$ is a canonical form then $\parent(o)$ is $\bot$ or a canonical form. 
	\end{enumerate}
\end{proposition}
\begin{proof}
	(1) If $o \in \Atoms$ then $n = 1$, if $o$ is a set/tuple then $n$ is the number of elements in the set/tuple.
	
	(2) Assume that there is $o \in \Cn$ and $o' = \parent(o) \neq \bot$ and there
  is $\pi \in \Perms$ such that $o'^\pi < o'$. Since $o' \neq \bot$, $o$ has to
  be a non-empty tuple or set by definition of $\parent$. If $o$ is a tuple then
  from the lexicographic ordering of tuples follows that $o^\pi < o$ and this is
  a contradiction. Now we explore the case $o = \{ o_1, \dots o_{n+1} \}$ where $o_i < o_j$ for $i
  < j$ and $i, j \in \{1, \dots, n+1 \}$. If there is $\{p_1, \dots p_{n}\} = o'^{\pi}$ such that $p_i < p_j$ for $i,j \in \{1, \dots, n\}$ then from fact that $o'^\pi < o'$ follows that there has to be $f \in \{1, \dots, n\}$ such that $p_f < o_f$ and $p_i = o_i < p_f$ for all $i \in \{1, \dots, f - 1\}$.
	The last step is to explore what happens when $\pi$ is applied on $o$; let $\{q_1, \dots q_{n+1}\} = o^{\pi}$ such that $q_i < q_j$ for $i,j \in \{1, \dots, n + 1\}$ and let $k \in \{1, \dots, n + 1\}$ such that $q_k = o_{n+1}^\pi$.
	Since applying $\pi$ on an object is bijective, $k \neq f$. 
	If $f < k$ then it follows that
	$o_i = p_i = q_i$ for $i \in \{1, \dots, f - 1 \}$ and $q_f = p_f < o_f$ and hence $o^\pi < o$.
	If $k < f$ then
	$o_i = p_i = q_i$ for $i \in \{1, \dots, k - 1 \}$ and $q_k < q_{k+1} = p_k = o_k$ and hence $o^\pi < o$.	
	\qed
\end{proof}

Proposition \ref{p:parent}.1 shows that $\parent$ defines a tree where:
$\bot$ is the root; non-root nodes are elements from $\Objs$;
\ref{p:parent}.2 shows that each canonical form can be reached from the root by a path that contains only canonical forms.
Moreover, the elements ``grow'' with the distance from the root.

This serves as a basis for the algorithm generating canonical forms
of elements from a domain. It recursively takes an object and tries to create a
bigger one, starting from $\bot$. On each level, it checks whether the new
element is canonical, if not, the entire branch is terminated. The way of getting a bigger object from a smaller one,     
depends on the specific domain, what type of objects are generated and by which
elements the already found elements are extended.
Since domains are composed from smaller ones, \Haydi{} iterates the elements of a subdomain to gain possible ``extensions'' to create a new object; such extensions are then added to the existing object to obtain a possible continuation in the tree.
Since subdomains are also strict domains, only through canonical elements of the
subdomain is iterated and 
new extensions are created by applying permutations on the canonical forms. 
Once it is clear that the extension leads to an object with a gap, then such a
permutation is omitted. Therefore, it is not necessary to go through all of the permutations.

Example:

\begin{minted}{python}
>>> a = hd.USet(1000, "a")
>>> b = a * a
<Product size=1000000 {(a0, a0), (a0, a1), (a0, a2), (a0, a3), ...}>
>>> list(b.cnfs())
[(a0, a0), (a0, a1)]
\end{minted}

The domain in variable \texttt{b} has one million of elements; however, only two 
of them are canonical forms. \Haydi{} starts with an empty tuple, then it asks for canonical forms of the subdomain \verb|a| that is a set containing only \texttt{a0}. The only permutation on \texttt{a0} that does not create a gap after adding into empty tuple is identity,
so the only relevant extension is \texttt{a0}. Therefore, only \verb|(a0,)| is
examined as a continuation. It is a canonical form, so the generation continues. Now the second domain from Cartesian product is used, 
in this particular example, again canonical forms of \verb|a| is used. At this
point the only no-gap (partial) permutations are identity and swap of
\texttt{a0} and \texttt{a1}, hence possible extensions are \verb|a0| and
\verb|a1|. Extending \texttt{(a0,)} give us \texttt{(a0, a0)} and \texttt{(a0,
  a1)} as results.

The approach is similar when sets are generated. The only thing that needs to be
added for this case is a check that the extending object is bigger (w.r.t. $\leq$) than previous
ones, to ensure that the current object is the actual parent of the resulting object.

\section{Distributed computations}\label{s:distcomp}

\Haydi{} was designed to enable parallel computation on cluster machines from the beginning.
\toolname{Dask/distributed}\furl{https://github.com/dask/distributed} serves as the backend for computations.
The code that uses this feature was already shown at the end of Section~\ref{s:ex-dg}.

\Haydi{} contains a scheduler that dynamically interacts with \toolname{dask/distributed} scheduler.
\Haydi{}'s scheduler gradually takes elements from a pipeline and assigns them to \toolname{dask/distributed}.
\Haydi{} calculates an average execution time of recent jobs and the job size is altered to having neither too small jobs
nor too large with respect to job time constraints.

\Haydi{} chooses a strategy to create jobs in dependence on a chosen method of a domain exploration. 
The simplest strategy is for randomly generated elements; the stream in the pipeline induces independent jobs and the scheduler has to care only about collecting results and adhering to a time constraint (that may be specified in \verb|run| method). 

In the case of iterating over all elements, there are three supported strategies: strategy for domains that support \emph{full slicing}, for domains with \emph{filtered slicing}, and a generic strategy for domains without slicing. 
The last one is a fallback strategy where \Haydi{} scheduler itself generates elements, these elements with the rest of the pipeline
are sent as jobs into \toolname{dask/distributed}.

The full slicing is supported if the number of elements in a domain is known and an iterator over the domain that skips the first $n$ elements can be
efficiently created.
In this case, the domain may be sliced into disjunct chunks of arbitrary sizes. \Haydi{} scheduler simply creates 
lightweight disjoint tasks to workers in form ``create iterator at $i$ steps and process $m$ elements'' without transferring 
explicit elements of domains. All built-in domains support full slicing as long as the filter is not applied.

If a domain is created by applying a filter, both properties are lost generally,
i.e., the exact number of elements, 
and an efficient iterator from the $n$-th item. However, if the original domain
supports slicing it is possible to utilize this fact. 
Domains can be sliced as if there was no filter present in domain or subdomains
at all, 
while allowing to signalize that some elements were skipped. Note, the filtered
elements cannot be silently swallowed, because the knowledge of how many elements
were already generated in the underlying domain would be lost. In such a case it
could not be possible to ensure that the iterations go over disjunct chunks of a domain.
The iterators that allow to signalize that one or more elements were internally
skipped are called ``skip iterators''.
The ability to signalize skipping more elements at once allows to implement efficient slicing when
filter domains are used in composition. For example, assume a Cartesian product of two filtered domains, where consecutive chunks of elements are dropped when a single element is filtered in a subdomain. This strategy usually works well in practice for domains
when elements dropped by the filter are spread across the whole domain. 

If canonical forms are generated, then the goal is to build a search tree. One job assigned to \toolname{dask/distributed} 
represents a computation of all direct descendants of a node in a search tree.
In the current version, it is quite a simple way of distributed tree search and there is a space for improvements; it is the youngest part of \Haydi{}.

\section{Performance}\label{s:perf}

The purpose of this section is to give a basic impression of \Haydi{}'s
performance. For comparison the two examples from
Sections~\ref{s:ex-dg} ad \ref{s:ex-cc} are used.

First of all, Haydi is compared with other tools. This comparison is
demonstrated on the example of generating directed graphs.
\toolname{Geng}~\cite{McKayGeng} is used as a baseline within the comparison,
since it is a state of the art generator for graphs. In order to simulate a
prototyping scenario a special version is included. This version loads graphs
generated by \toolname{geng} into Python. The loading process is done using the \toolname{networkx}\furl{https://networkx.github.io/} library and a small manual wrapper.
Moreover, the results of two other experiments are included. Both experiments summarize
graphs generated by \toolname{SageMath}, in the first case \toolname{SageMath}
uses \toolname{geng} as the backend while in the other one it uses its own graph
implementation \toolname{Cgraph}.
\Haydi{} was executed with Python 2.7.9 and PyPy 5.8.0. \toolname{Geng} 2.6r7
and \toolname{SageMath} 8.0 were used. Experiments were executed on a laptop
with Intel Core i7–7700HQ (2,8 GHz). Source codes of all test scripts can be
found in the \Haydi{}'s git repository. The results are shown in Table~\ref{tab:geng}.
In all cases, except the last one, the goal was to generate all non-isomorphic graphs  with the given number of vertices without any additional computation on them. The last entry generates all possible graphs (including isomorphic ones) and runs in parallel on 8 processes. 

\begin{table}
\caption{Performance of generating graphs}\label{tab:geng}
\centering
\begin{tabular}{ |c|c|c|c|c| }
	\hline
	\bf{Tool/\# of vertices}					& \bf{5}    	& \bf{6}    & \bf{7}    & \bf{8} \\
	\hline
	\toolname{geng} (without loading to Python) 						& \textless0.01s & \textless0.01s & \textless0.01s  & 0.01s		\\
	\hline
	\toolname{geng} + manual parser 		& \textless0.01s & 0.03s   & 0.05s    & 0.11s		\\
	\hline
	\toolname{geng} + \toolname{networkx} 			& 0.27s   & 0.28s   & 0.28s    & 0.94s		\\
	\hline
	\toolname{SageMath} (\toolname{geng} backend) 	& 0.17s 	 & 0.18s   & 0.26s    & 2.17s    	\\
	\hline
	\toolname{SageMath} (\toolname{Cgraph} backend) 	& 0.09s 	 & 0.55s   & 6.59s    & 139.76s  	\\
	\hline
	\Haydi{} canonicals (Python)	& 0.53s 	 & 11.55s  & timeout & timeout 	\\
	\hline
	\Haydi{} canonicals (PyPy)		& 0.34s	 & 5.46   & timeout & timeout	\\
	\hline
	\Haydi{} parallel \texttt{iterate()} (PyPy)	& 3.27s  & 3.49s   & 70.99s   & timeout	\\
	\hline
\end{tabular}\\
timeout is 200s
\end{table}

It is obvious that Haydi cannot compete with Geng in generating graphs. Geng is hand-tuned for this specific use case,
in contrast to Haydi that is a generic tool. On the other hand, the Haydi program that generates graphs can be simply extended or modified to generate different custom structures while modifying Geng would be more complicated.



The second benchmark shows strong scaling of parallel execution of the reset word generator from Section~\ref{s:ex-cc} for six vertices and two alphabet characters for variants where \verb|cnfs()| was replaced by \verb|iterate()|, since the parallelization 
of \verb|cnfs()| is not fully optimized, yet. The iterated domain supports the full slicing mode.
The experiment was executed on the Salomon cluster.

\begin{table}
\small
\caption{Performance of \texttt{iterate()} on Salomon}\label{tab:haydi-dist}
\centering
\begin{tabular}{ |c|c|c| }
	 \hline
	 \bf{Nodes (24 CPUs/node)} & \bf{Time} & \bf{Strong scaling} \\
	 \hline
	 1 & 3424s & 1 \\
	 \hline
	 2 & 1908s & 0.897 \\
	 \hline
	 4 & 974s & 0.879 \\
	 \hline
	 8 & 499s & 0.858 \\
	 \hline
\end{tabular}
\end{table}

The last note on performance: we have experimented with several concepts of 
the tool. The first version was a tool named \toolname{Qit}\furl{https://github.com/spirali/qit},
that shares similar ideas in API design with \Haydi{}. It also has Python API, but generates 
C++ code behind the scene that is compiled and executed.
Benchmarks on prototypes showed it was around $3.5$ times faster than pure Python version (executed in PyPy);
however, due to C++ layer, \toolname{Qit} was less flexible than the current version \Haydi{} and hard to debug for the end user.
Therefore, this version was abandoned in favor of the pure-python version to obtain a more flexible environment for experiments and prototyping. As the problems encountered in generation of combinatorial objects are often exponential, $3.5\times$ speedup does not compensate inflexibility.

\bibliographystyle{splncs03}
\bibliography{references}

\clearpage
\appendix

\section{Function \texttt{check\_automaton}}
\label{A:check-automaton}

\begin{minted}{python}
from haydi.algorithms import search

# Let us precompute some values that will be repeatedly used
init_state = frozenset(states)
max_steps = (n_states**3 - n_states) / 6
# Known result is that we do not need more than (n^3 - n) / 6 steps

def check_automaton(delta):
	# This function takes automaton as a transition function and
	# returns the minimal length of reset word or 0 if there
	# is no such word

	def step(state, depth):
		# A step in bread-first search; gives a set of states
		# and return a set reachable by one step
		for a in alphabet:
		    yield frozenset(delta[(s, a)] for s in state)

	delta = delta.to_dict()
	return search.bfs(
		init_state,  # Initial state
		step,        # Function that takes a node and
                             # returns the followers
		lambda state, depth: depth if len(state) == 1 else None,
		# Run until we reach a single state
		max_depth=max_steps,  # Limit depth of search
		not_found_value=0)    # Return 0 when we exceed
                                      # depth limit
\end{minted}

\end{document}